\documentclass[10pt]{llncs}

\usepackage{graphicx}
\usepackage{algorithmic}
\usepackage{algorithm}

\title{Distributed Algorithm for Dynamic Data-Gathering in Sensor Network}
\author{Subhasis Bhattacharjee}
\institute{University of Bristol\\
Merchant Venturers Building\\
Bristol, BS8 1UB\\
\email{subhasis18@gmail.com}}

\usepackage[%
  pdftitle={},%
  pdfauthor={Subhasis},%
  pdfsubject={},%
	pdfcreator={Microsoft® Word 2010},%
  pdfproducer={Microsoft® Word 2010},%
  pdfkeywords={keyword1,keyword2}]{hyperref}
\begin{document}

\maketitle

\begin{abstract}

In wireless sensor networks (WSN), each sensor is responsible for sensing environmental conditions and sending them to the one or more base stations. Battery-operated sensors are severely constrained by the amount of energy that can be spend for transmitting these sensed data. However, aggregation of data (including removal of redundant data) at intermediate sensors and forwarding of aggregate data reduce overall energy consumptions in WSN. In general, data gathering refers to the process of periodic collection of sensed data from various sensors to one or more base stations (BS). Energy efficient data gathering scheduling is essential for improving the lifetime of WSN. In this paper, we propose a distributed algorithm to compute data-gathering schedule that aim to improve the lifetime of WSN by suitably selecting energy-efficient data-flow paths from various sensors to the base station. For a multihop WSN with $n$ sensors, the proposed algorithm first computes a schedule in $O(n^2)$ time steps, and then this schedule is periodically updated based the residual energy and the feedback received from the BS. The system performs approximately $\log(\mathcal{L})$ schedule updates where $\mathcal{L}$ is the expected lifetime of the system in number of data-gathering rounds. Moreover, each updation process uses the existing active schedule (data-flow path) - thus consuming only a small fraction of a single data gathering round activity. Such an algorithm thus could precisely incorporate the energy consumptions due to updates and related activities. Moreover, our algorithm does not assume any global knowledge of the topology or the positions of various sensors. Through simulation study, we found that our proposed algorithm achieves significantly higher network lifetime compared to existing data-flow schedules based on the Minimum Spanning Tree (MST), the Shortest Path Tree (SPT), the Weighted Rooted Tree (WRT) \cite{Bhattacharjee07}.
\end{abstract}

\section{Introduction}

A typical Wireless Sensor Network (WSN) consists of a large number of inexpensive sensor nodes, densely deployed over an area to gather information about the surroundings in a self-configured and unattended fashion for a long period of time. The basic operation in WSN is the periodic sensing, gathering and transmission of data by individual sensors to one or more base stations (BS). Since sensor nodes are battery driven and there is no facility for recharging in general, nodes must route data in most energy-efficient way. The energy of a node is drained out with time for data sensing, computing and communicating with other nodes and the base station. Following conventional communication model, we see that the amount of energy spent in transmitting a packet has a fixed cost in electronics and a variable cost that depends on the distance between the communicating nodes, whereas, the reception of a data packet has a fixed energy cost in electronics. In WSN, data from all nodes are collected and transmitted to the BS in periodic intervals called rounds. A {\it data gathering schedule} specifies how data packets are routed and collected at the base station in each round. According to the most restricted definition, the {\it lifetime} of a WSN is defined to be the time duration (i.e., the number of rounds) for which the base station can gather data from the sensors.

A simple approach is to transmit the sensed data directly to the BS. The BS being located far away, the cost of transmitting directly to the BS from any node can be significantly high. Moreover, it may not be feasible when sensors are deployed over a large region. So, it is better to route data in multihop paths via intermediate nodes. Here, data fusion \cite{Li08} or aggregation helps in reducing the amount of data flow. The study on data fusion or aggregation essentially shifts the focus to data centric approaches \cite{Intanagonwiwat00}, \cite{Krishnamachari02}, \cite{Lindsey01}, \cite{Lindsey02}, \cite{Madden02} for reducing traffic in WSN. Usually, the gathered/aggregated data move from node to node, get fused, and eventually received at the BS. To maximize network lifetime, energy-efficient data gathering schedules are of primary interest in a successfully designed WSN. However, the problem of finding a data gathering schedule that maximizes the lifetime is in general NP-hard. Moreover, the challenge here is to have a reasonably good data gathering schedule that can also adapt with changes in the topology due to node failure.

The problem of finding an efficient data gathering algorithm that maximizes the lifetime, referred to as the Maximum Lifetime Data Aggregation (MLDA) problem \cite{Intanagonwiwat00}, \cite{Kalpakis03} is in general NP-hard \cite{Ramanathan00}. In \cite{Kalpakis03}, the problem is mapped to integer linear programming and later solved in polynomial time by linear relaxation to produce a linear approximation to the optimal solution.

Extensive research has been done so far on maximizing the lifetime of wireless sensor networks \cite{Song09},\cite{Wang08}. Several protocols attempt to improve lifetime by reducing energy consumption in MAC for WSN \cite{Ye02}, \cite{Yu07}. In \cite{Schurgers02} the lifetime of WSN is enhanced by efficient use of MAC. It reduced MAC overhead in WSN by using STDMA and an encoded representation of the addresses in data packets. It is shown by simulation that the MAC overhead is reduced by a factor of three compared to the existing approaches reducing the energy consumption per transmission and reception.

The other approach to improve the lifetime is by partitioning the nodes into two sets, one as operational ({\it awake}) and the other as backup ({\it asleep}) \cite{Cerpa02}, \cite{Xu01}, \cite{Chen02}, \cite{Hong08}, \cite{HongTA}. However, such partitioning requires knowledge of their positions using GPS or similar devices. 

Some study \cite{Lin97}, \cite{Banerjee07L} have also been carried out using hierarchical routing based on clustering to improve lifetime. The HEED protocol \cite{Younis04} proposes a distributed clustering approach that enhances the lifetime by distributing energy consumption among the cluster head nodes, terminating within a constant number of steps and considering a combination of energy and communication cost for selecting cluster heads. In LEACH \cite{Heinzelman00}, a distributed data gathering procedure is proposed using clustering based protocol. PEGASIS \cite{Lindsey02} improved the performance by forming a chain with the sensors where each sensor communicates with the base station in turn to deliver the aggregated data. But all these works assumed a single-hop WSN (with global knowledge of topology), where each sensor can communicate directly with each other and also with the base station. A distributed approach for computing data-gathering schedule in multihop WSN is presented in \cite{Bhattacharjee07}. It does not assume any global knowledge of the topology. It performs better than any data-gathering schedule generated from shortest path tree and minimum spanning tree (MST). Moreover, it also performs better than PEGASIS \cite{Lindsey02} in general. It is worthwhile to note that the power assignment following a MTS in a graph achieves $2$-approximation in terms of optimal power assignment \cite{Li05}. Though we are not aware of any similar result applicable to data-gathering in WSN, the data-gathering schedule produced out of a MST is in general an energy-efficient one.

Many energy-efficient routing protocols have been proposed so far that use cluster-based routing \cite{Dimokas10} \cite{Dahnil12} \cite{Aslam11} \cite{Chamam09} in WSN. Here high energy is consumed for inter-cluster communication by the cluster-heads. So many of these schemes are adaptive in general and recomputes cluster heads periodically - but there is no systematic study on exact number of re-computations, the exact energy and time overhead. As oppose to that, in this paper, we give precisely the exact number of reschedules, their overhead and delay penalty.

The paper is organized as follows. Section~\ref{sec:SystemModel} presents the system model. The proposed algorithm is presented in Section 3. In Section 4, we describes the operation of the WSN and dynamic updates of the data gathering schedule. The simulation results are presented in Section 5. We conclude our paper in Section 6.

\section{System Model}\label{sec:SystemModel}

We assume that the system initializes with a set of $n$ sensors $\{v_1, v_2, \cdots, v_n\}$, randomly distributed over a region along with a fixed base station $BS$ as a special node. The BS can transmit directly to all sensors. But each sensor $v_i$ is capable of transmitting a packet within its limited range $R_i$. So, the nodes within range $R_i$ are adjacent neighbors of $v_i$ and are reachable from $v_i$ in single hop. The rest of the nodes can be reached in multihops via intermediate nodes. The maximum transmission power $P_{max}$ of a sensor node covers the maximum transmission range $R_{max}$ ($P_{max}$ is assumed constant). This sensor network can be represented as a graph as defined below.

\begin{definition}
The topology graph $G(V, E)$ representing a sensor network is an undirected graph with a node set $V = \{v_1, v_2,...v_n, BS\}$, where  two nodes are adjacent if and only if the Euclidean distance between the two is less than $R_{max}$.
\end{definition}

Using the popular existing energy consumption model \cite{Heinzelman00}, \cite{Lindsey01}, the energy consumed by a sensor $v_i$ in receiving a $k$-bit message is $$Rx = \epsilon_{elec} \times k.$$ The energy consumed by sensor $v_i$ to transmit a $k$-bit message to $v_j$ is $$Tx_{i,j} = Rx + \epsilon_{amp} \times k \times d_{i,j}^2,$$ where $\epsilon_{elec}$ is the energy required by the transmitter or receiver circuit and $\epsilon_{amp}$ is that for the transmitter amplifier to transmit single bit, and $d_{i,j}$ is the Euclidean distance between $v_i$ and $v_j$. We assume that the radio channel is symmetric (i.e., the energy consumed to transmit a message from node $v_i$ to node $v_j$ is same as in the reverse direction). 

Each sensor $v_i$ is capable recording its available (residual) energy $RE_i$ - which is same for all nodes at the beginning. However, as the WSN performs periodic data gathering - the available energy reduces proportionately in each node over time. 

We assumed that the system is initially connected and so, $G(V,E)$ is connected. The the sensors are assumed to be static. $G(V,E)$ is undirected since the transmission medium under consideration is assumed to be symmetric.

\begin{definition}
A weighted topology graph $G(V, RE, E, W)$ is the topology graph $G(V, E)$ with $RE_i$ the residual energy of each node $v_i$ and the weight $w_{i,j} \in W$ for each edge $(v_i, v_j) \in E$, where $w_{ij} = Tx_{i,j}$ represents the amount of energy required to transmit a data packet of $k$ bits from $v_i$ to $v_j$ and vice-versa.
\end{definition}

The proposed algorithm (in Section~3) computes a {\it rooted tree} that can be used as a data gathering schedule where the BS is the root node. The goal of our algorithm is to maximize the data gathering rounds (simply, the number of rounds) based on energy consumption in message transmission and reception at each sensor node for the underlying weighted topology graph $G(V, RE, E, W)$ corresponds to the given WSN.

We denote a {\it rooted tree} by $T(v_t, V_T, RE_T, E_T)$, where, $v_t$ is the root of the tree, $V_T$ is the set of nodes, $RE_T$ is residual energy of nodes, and $E_T$ is the set of directed edges of the tree. A {\it weighted rooted tree} (WRT), denoted as $T(v_t, V_T, RE_T, E_T, W_T)$, is a rooted tree where $w_{ij} \in W_T$ is the weight of the directed edge $ {(v_i, v_j)} \in E_T$ that represents the transmission power required for the link $ {(v_i, v_j)}$.

\begin{definition}
For a weighted rooted tree $T(v_t, V_T, RE_T, E_T, W_T)$, the {\it expected lifetime ($\mathcal{L}$)} for each node $v_i \in V_T-\{v_t\}$ is $\mathcal{L}_i = \frac{RE_i}{d_{in} \times Rx + w_{i,out(v_i)}}$, where $d_{in}$ denotes the in-degree of $v_i$ and $out(v_i) \in V_T$ is the node to which $v_i$ forwards its data. The {\it minimum expected lifetime} is defined as $\mathcal{L}_{min} = \textrm{ min }\{\mathcal{L}_i | \forall v_i \in V-\{v_t\}\}$.
\end{definition}

Therefore, we restate the lifetime improvement problem as the problem of finding a rooted tree $T(v_t, RE_T, V_T, E_T', W_T')$, (i.e., rooted at $BS$ denoted as $v_t$), such that the minimum expected lifetime $\mathcal{L}_{min}$ is maximized, where, $E' \subseteq E$, and $w'_{i,j} = w_{i,j}, \forall (v_i, v_j) \in E'$. Then, the nodes can follow the scheduling determined by $T(v_1, V, E', W')$.

\section{Distributed Algorithm for Data-Gathering Schedule}

We do not assume any global knowledge of the topology. We just assume that each node knows the information about its (immediate / one-hop) neighbors (their residual energy, weights, etc). Let $N_i$ denotes the set of nodes adjacent to $v_i$. The algorithm starts by considering the node $BS$ as the only member of WRT $T^0$. In $k$th iteration step the tree is $T^k(BS, V^k, E^k, W^k)$. In $k+1$th iteration, a node $v_i \not\in V^k$, but is adjacent to some node $v_j \in V^k$, is included in $T^k$ such that the $\mathcal{L}_{min}$ is maximal within the set of nodes $V^k \cup \{v_j\}$ (i.e., after inclusion of the node $v_j$). This process continues until $T$ covers all $n$ nodes. A formal presentation of the algorithm is given below.

\begin{algorithm}[htb]
\begin{algorithmic}
\STATE{{\it Initialization}}
	\STATE{$T^0 = (BS, \{BS\}, \phi, \phi, \phi)$}
	\STATE{}
\STATE{{\it In Step $k$}}
	\FOR{each $v_i \in V^{(k-1)}$}
	\STATE{Find $v_j$ such that $\mathcal{L}_j$ is maximum $\forall v_j \in N_i$}
	\STATE{Compute $\mathcal{L}_i = \frac{RE_i}{(d_{in}+1) \times Rx + w_{i,out(v_i)}}$}
	\IF{$d_{in} > 0$}
		\STATE{Wait to collect all $(v_j, \mathcal{L}_i, v_{j'}, \mathcal{L}_{j'})$ from in-neighbors of $v_i$}
		\STATE{Compute max tuple $(v_m, \mathcal{L}_m, v_{m'}, \mathcal{L}_{m'})$ from received tuples and its own}
		\STATE{This $(v_m, \mathcal{L}_m, v_{m'}, \mathcal{L}_{m'})$ corresponds to maximum of the paired tuples}
	\ENDIF
	\STATE{Send $(v_m, \mathcal{L}_m, v_{m'}, \mathcal{L}_{m'})$ message to $out(i)$ in $T^{k-1}$}
	\IF{the message $v_i, v_k)$ is received from BS}
		\STATE{Mark $v_k$ as in neighbor of $v_i$ in $T^k$}
		\STATE{Add $(v_k, v_i)$ in $E^k$}
	\ENDIF
\ENDFOR
\FOR{each node $v_j \not\in T^{(k-1)}$}
	\IF{a message is received from BS}
		\STATE{$v_i$ is included in $V^{k}$}
	\ELSE
		\STATE{Wait for message from BS}
	\ENDIF
\ENDFOR
\STATE{}
\STATE{{\it In Step $k$ for $BS$}}
\IF{received $(v_j, \mathcal{L}_i, v_{j'}, \mathcal{L}_{j'})$ message}
\STATE{Wait to collect all $(v_j, \mathcal{L}_i, v_{j'}, \mathcal{L}_{j'})$ from in-neighbors of BS}
\STATE{Compute max tuple $(v_m, \mathcal{L}_m, v_{m'}, \mathcal{L}_{m'})$ from received tuples}
\STATE{Broadcast $(v_m, v_{m'})$ message}
\ENDIF
\STATE{Record min expected lifetime $\mathcal{L}_{min}$ for future use}
\end{algorithmic}
\caption{ComputeFullSchedule} \label{algo:ComputeFullSchedule}
\end{algorithm}

\begin{figure}[!ht]
\centerline{\includegraphics[height=2in]{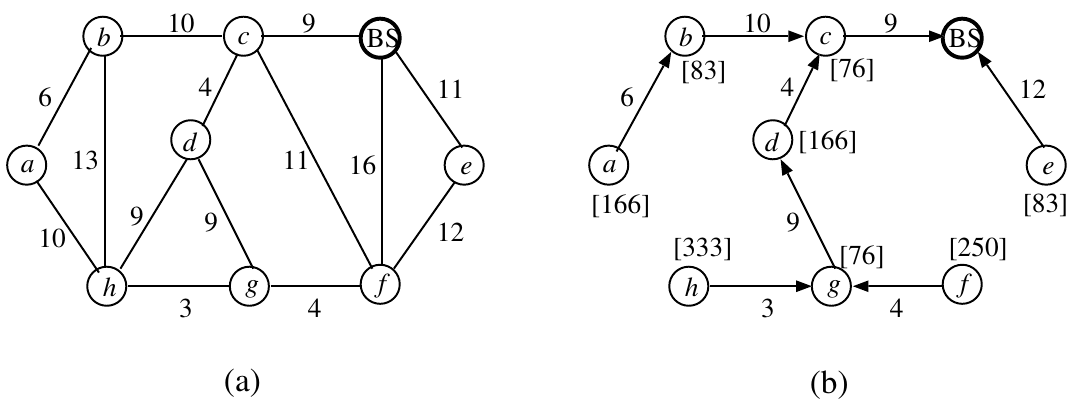}}
\caption{(a) Weighted topology graph for WSN, and (b) Scheduled Rooted Tree}
\label{fig:sensor:wrt}
\end{figure}

Figure~\ref{fig:sensor:wrt}(a) shows a WSN with $8$ sensor nodes and one BS. The numbers on the edges are the weights ($w_{i,j}$) of the edges. Figure~\ref{fig:sensor:wrt}(b) shows the generated rooted tree $T$ with the node $BS$ as the root corresponding to $G(V, RE, E, W)$ shown in Figure~\ref{fig:sensor:wrt}(a). Assuming $Rx = 2$, and $RE_i = 1000, \forall v_i \in V$, the numbers in bracket shows the expected lifetime $\mathcal{L}_i$ for each node $v_i$. Note that here $\mathcal{L}_{min} = 76$.

\begin{theorem}
The ComputeFullSchedule algorithm terminates in $O(n^2)$ steps.
\end{theorem}

\begin{proof}
It is evident from the steps of the algorithm.
\end{proof}

However, this $O(n^2)$ bound is applicable if only if the generated tree is actually a chain of $n$ nodes, each having exactly one incoming edge and one outgoing edge (except the two end points, one being the BS). In general, the rooted tree generated as a data gathering schedule in WSNs with uniform distribution of sensors contains intermediate nodes with in-degree greater than one. Moreover, the algorithm does not use any long-distance transmission from any sensor to the BS - and it uses only one short distance transmission (per round) to its immediate out-neighbor in the partially formed rooted tree (even while constructing the very first data gathering schedule). This process thus saves energy to a great extend for each node. It is worthwhile to note that this computation is done only once for the whole lifetime of the WSN. In the next section we demonstrate the dynamic adjustment algorithm that re-constructs new rooted tree from the existing one with very little energy consumption (overhead) and delay. 

\section{Dynamic Schedule Computation}

During the operation of WSN, it periodically performs data gathering in rounds. Each sensor generates one data packet in each data-gathering round to be transmitted to BS via data gathering schedule. We use the data aggregation model/technique presented in \cite{Lindsey02}. Here, a node $v_i$ receives data packets from its down-stream neighbors (i.e., set of sensors as in-neighbors of $v_i$ in the rooted tree), each is of same size, say $k$ bits. It fuses all incoming data as well as its own sensed data, and forwards the aggregated information as a single data packet of size $k$ bits to the up-stream neighboring sensor $out(i)$ (for forwarding it to the base station).

The BS keeps track of the number of rounds the WSN is performing the data gathering operation. Moreover, BS gathers the knowledge of the longest path $\kappa$ (in number of hops) from any leaf sensor node to the BS. This information can be easily gathered using piggybacking with the normal sensed data - this details are omitted here. BS also stores the minimum expected lifetime $\mathcal{L}_{min}$ of the system as it came to know during the computation of the first schedule. As soon as BS observes that it has reached the round corresponds to $\frac{\mathcal{L}_{min}}{2} - \kappa$, BS broadcast a RESCHEDULE message with $\kappa$ as a parameter. This informs the sensors to start rescheduling operation and piggyback the re-scheduling related data-values with the sensed data. (Note that each sensor recomputes the $\mathcal{L}$ values taking into account of next $\kappa$ rounds - this details can be figured out easily.) The rescheduling finishes exactly after $\kappa$ rounds and then sensors move into the new schedule corresponds to the updated rooted tree - and the normal operation proceeds as usual. A formal presentation of the algorithm is given below.

\begin{algorithm}[htb]
\begin{algorithmic}
\STATE{{\it In Each Round $r$}}
\FOR{each $v_i \in V$}
	\IF{RESCHEDULE message received from $BS$}
		\STATE{Set $ReCompRound = 1$}
		\STATE{Re-Compute $\mathcal{L}_i$ and max tuple $(v_m, \mathcal{L}_m, v_{m'}, \mathcal{L}_{m'})$}
		\STATE{Piggyback {\bf this} with the sensed data}
		\IF{Message $(v_i, v_k)$ is received from $BS$}
			\STATE{Mark $v_k$ as in neighbor of $v_i$ in $T^k$}
			\STATE{Add $(v_k, v_i)$ in $E^k$}
		\ENDIF
	\ENDIF
	\IF{$ReCompRound == \kappa$ (i.e., $\kappa$ rounds passed)}
		\STATE{Switch to new Scheduling Tree}
	\ELSE 
		\STATE{Do sensing, data-fusion, aggregation and send data to $out(i)$}
	\ENDIF
\ENDFOR
\STATE{}
\STATE{{\it For $BS$}}
\IF{$(\mathcal{L}_{min} > \kappa)$ and $(r == \frac{\mathcal{L}_{min}}{2} - \kappa$)}
	\STATE{Broadcast (RESCHEDULE, $\kappa$) message}
\ENDIF
\IF{RESCHEDULE happened in round $(r-1)$}
	\STATE{Record the new $\mathcal{L}_{min}$}
	\STATE{Set $r = 0$}
\ENDIF
\end{algorithmic}
\caption{DataGatheringReScheduling} \label{algo:DataGatheringReScheduling}
\end{algorithm}

It should be noted here that re-scheduling operations are performed within the scope of normal data gathering round - and thus this process does not incur any extra delay. The energy overhead in sending these newly computed values (for re-scheduling) is very minimal as they are being appended with the normal traffic and occupy only a very small fraction of the data aggregation traffic. Let $\mathcal{L}_{min}$ denotes the minimum expected lifetime (in rounds) as computed by the ComputeFullSchedule algorithm. The number of rescheduling is close to $\log(\mathcal{L}_{min})$. It is an interesting combinatorial problem to find out the exact bound for this - and we are still working on it. However, our simulation shows that this value is close to $\log(\mathcal{L}_{min})$. Moreover, whenever the expected lifetime of the WSN comes close to $\kappa$, it does not perform any rescheduling - which is important as the sensor network is supposed to expire even before this reschedule process ends.

{\it We call our proposed technique as {\bf Lifetime with Dynamic Schedule} method and we abbreviate it as {\bf LDS} for our future discussion.}


\section{Performance Evaluation}

We performed extensive simulation to compare the outcome of the proposed algorithm with the existing techniques. For comparison with earlier works, the values of $\epsilon_{elec}$, $\epsilon_{amp}$ and $k$ have been taken from \cite{Lindsey01} and are as follows:
\begin{itemize}
\item[] $\epsilon_{elec} = 50$ $nJ/bit$,
\item[] $\epsilon_{amp} = 100$ $pJ/bit/m^2$, and
\item[] $k = 2000$.
\item[] Initial energy = 0.25 Jules.
\end{itemize}

Random WSNs are generated with the number of sensors varying between $30$ and $200$. The sensor are distributed randomly over a $200 \times 200$ 2-D region. The $R_{max}$ is varied between $40$ and $100$. For each value of $n$, at least $100$ random graphs are generated. We compared the proposed algorithm with the data gathering techniques via minimum spanning tree (MST), shortest path tree (SPT), weighted rooted tree (WRT) schedule \cite{Bhattacharjee07}. Given a topology graph the MST algorithm generates a spanning tree that minimizes the total link cost. For SPT algorithm, the single sink/destination shortest path tree is computed using Dijkstra's algorithm, where the sink is the BS. We assumed that data gathering takes place from leaf nodes towards the root which is the base station. Further discussion can be found in \cite{Bhattacharjee07}.

\begin{figure}[!ht]
\centerline{\includegraphics[width=3.2in]{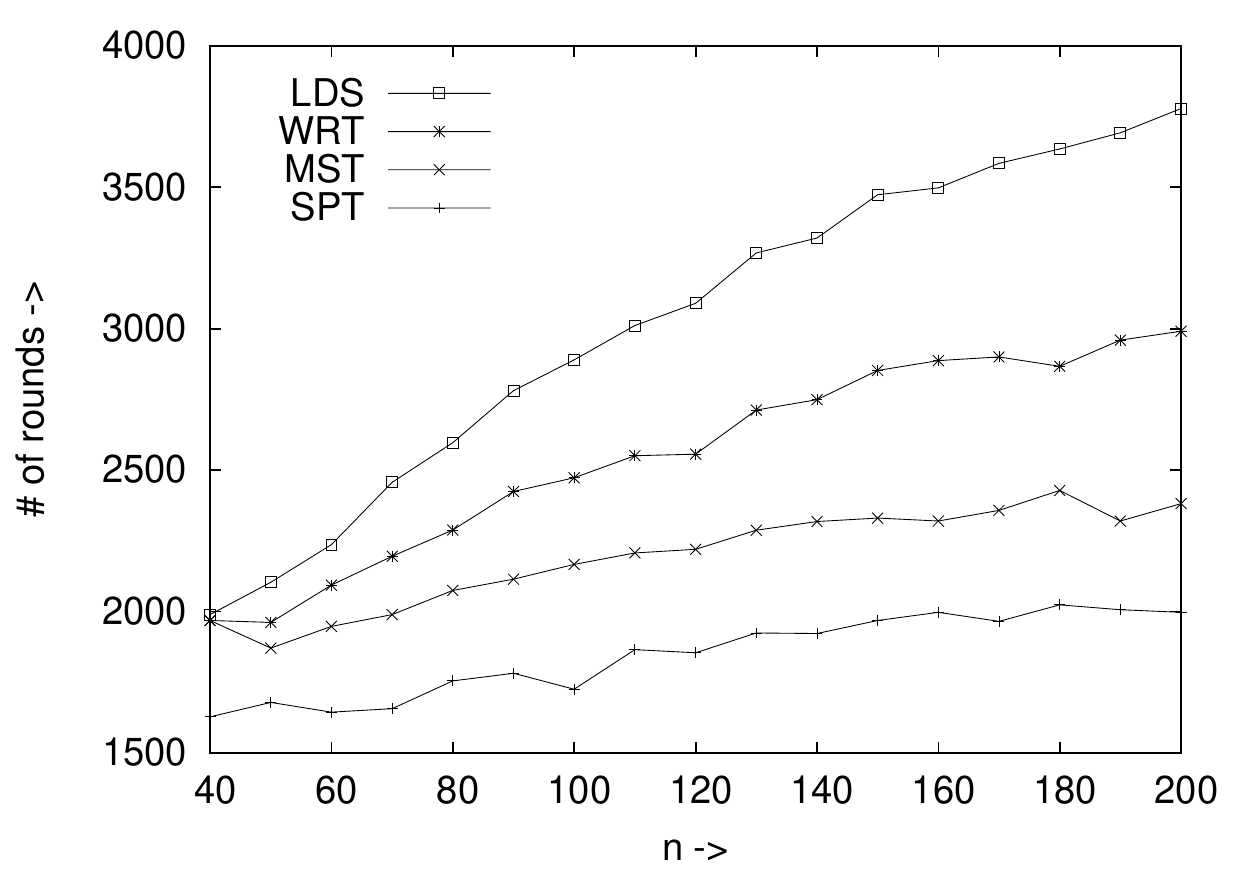}}
\caption{Lifetime (in rounds) vs number of node ($n$) for a fixed $range = 80$ units}
\label{Plot:fixRange}
\end{figure}

These three existing algorithms and our proposed algorithm (in this paper) are simulated on same graph and the corresponding lifetime values (in number of rounds) are recorded. Fig.~\ref{Plot:fixRange} shows the variation of lifetime with number of nodes with for a fix $range=80$ units and the number of nodes $n$ is varied from $40$ to $200$. We observe that our proposed Lifetime with Dynamic Schedule (LDS) method improves the lifetime by more than $25\%$ on an average compared to the existing algorithm (WRT) \cite{Bhattacharjee07} (for large sensor networks comprising $100$ or more sensors). We can readily observe from Fig.~\ref{Plot:fixRange} that it performs much much better than data-gathering schedule from MST and SPT.

\begin{figure}[!ht]
\centerline{\includegraphics[width=3.2in]{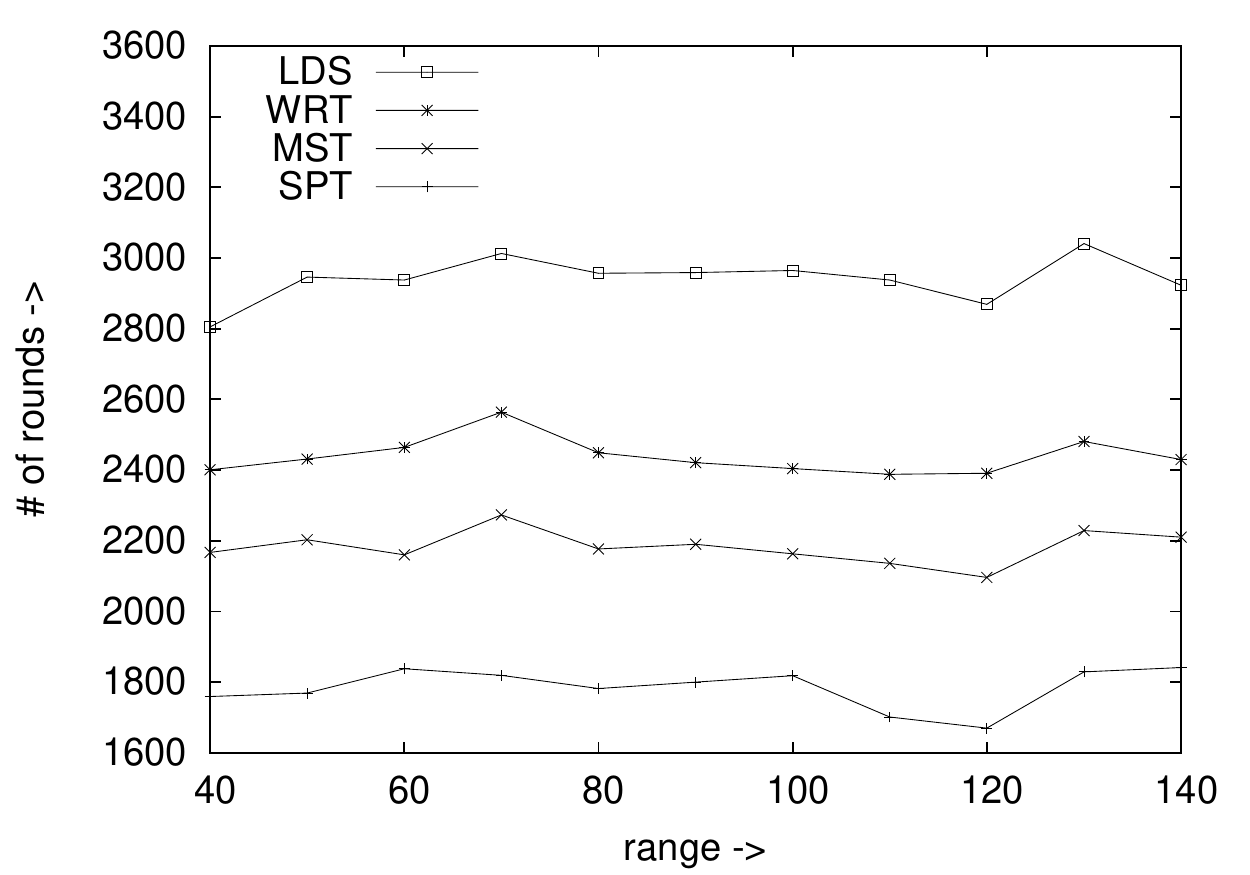}}
\caption{Lifetime (in rounds) vs maximum tranmission range ($R_{max}$)for a fixed $n=100$}
\label{Plot:fixN}
\end{figure}

Fig.~\ref{Plot:fixN} shows the variation of lifetime with the maximum transmission range of nodes. We kept the number of nodes $n = 100$. Here, $R_{max}$ is varied from $40$ to $100$. In all the cases, the proposed LDS method significant improves the lifetime of the WSN. 

In our simulation, we closely observed the average number of rescheduling performed during the lifetime of WSNs. We found that for all our simulation the number of rescheduling varies between $6$ to $11$. Moreover, in our simulation, we accounted the energy consumptions due to various rescheduling operations and first time schedule computation more accurately. Whenever there is a piggyback with the normal data we have added $10\%$ extra data and accounted the corresponding energy consumptions in our simulation.

\section{Conclusion}

We presented a distributed data gathering algorithm for improving the lifetime in WSN. Starting from a random distribution of $n$ sensors, the algorithm first computes a data gathering schedule. The expected lifetime (in number of rounds) of the WSN is computed during every rescheduling. During the course of normal data gathering, the base station initiates the rescheduling whenever half of the expected lifetime of the WSN is finished. This clever periodic re-scheduling (similar to half-life period) reduces the total number of reschedule operations but at the same time enhances the lifetime of WSN significantly. Though the computation of the first schedule may take a time proportional to $O(n^2)$ at the extreme case, the rescheduling procedure incurs no delay at all to the system. No knowledge of global topology is required for our computation. Simulation studies show that the proposed LDS algorithm improves the lifetime of WSN by more than $25\%$ compared to some of the best known existing data gathering algorithms. The proposed distributed algorithm can be  well-suited for large multihop WSNs.

\bibliographystyle{abbrv}
\bibliography{../../MyBiblio/SubhasisBRefConf,../../MyBiblio/SubhasisBRefJournal,../../MyBiblio/RefBookAdhocNetwork,../../MyBiblio/RefConfAdhocNetwork,../../MyBiblio/RefJournalAdhocNetwork}

\end{document}